\documentclass[a4paper,12pt]{article}
\usepackage[utf8]{inputenc}
\usepackage{graphicx}
\graphicspath{ {./images/} }
\usepackage[dvipsnames]{xcolor}
\usepackage[utf8]{inputenc}
\usepackage{cite}
\usepackage{geometry, amsthm, amsmath, amssymb, graphicx, float, enumerate, hyperref, longtable, braket, mathtools, bbm}
\usepackage[utf8]{inputenc}
\usepackage{autonum}
\usepackage{enumitem}
\allowdisplaybreaks % allows page breaks
\numberwithin{equation}{section}
\DeclareMathAlphabet{\mathpzc}{OT1}{pzc}{m}{it}

\restylefloat{table}
\restylefloat{figure} 

\swapnumbers
\newtheorem{theorem}{Theorem}[section]

\newtheorem{corollary}[theorem]{Corollary}

\newtheorem{lemma}[theorem]{Lemma}
\newtheorem{proposition}[theorem]{Proposition}
\theoremstyle{definition}
\newtheorem{definition}[theorem]{Definition}

\theoremstyle{remark}
\newtheorem{remark}[theorem]{Remark}

\newtheorem*{corollary-non}{Corollary}

\def\bC{\mathbb{C}}

\def\bN{\mathbb{N}}

\def\bR{\mathbb{R}}

\def\cB{\mathcal{B}}

\def\cD{\mathcal{D}}
\def\cH{\mathcal{H}}
\def\cI{\mathcal{I}}
\def\cK{\mathcal{K}}
\def\cL{\mathcal{L}}

\def\cP{\mathcal{P}}
\def\cQ{\mathcal{Q}}
\def\cR{\mathcal{R}}

\DeclareMathOperator*{\argmax}{argmax}
\title{Saturating the Data Processing Inequality for $\alpha-z$ R{\'e}nyi Relative Entropy}
\author{Sarah Chehade}
\date{}
\begin{document}

\maketitle
%\tableofcontents

\begin{abstract}
    
    It has been shown that the $\alpha-z$ R{\'e}nyi relative entropy satisfies the Data Processing Inequality (DPI) for a certain range of $\alpha$'s and $z$'s. Moreover, the range is completely characterized by Zhang in `20.  We prove necessary and algebraically sufficient conditions to saturate the DPI for the $\alpha-z$ R{\'e}nyi relative entropy whenever $1<\alpha\leq 2$ and $\frac{\alpha}{2}\leq z\leq\alpha$. Moreover, these conditions coincide whenever $\alpha=z$.  
     
\end{abstract}

%%%%%%     INTRODUCTION   %%%%%%%%
\section{Introduction}\label{sec: intro}

Statistical distinguishability between two states is a central concept in Quantum Information Theory. One basic distinguishability measure, the quantum relative entropy, was introduced by Umegaki in 1962 in his paper \cite{umegakiEntropy} about non-commutative conditional expectations. When two states pass through a noisy quantum channel, it is indeed more challenging to measure this distinguishability. This phenomena is described as the data processing inequality. In order for a distinguishability measure to have any operational meaning, it must satisfy the data processing inequality. Petz \cite{petz, petz1} proved this inequality in the context of von Neumann Algebras. More information about this relative entropy is found in section \ref{relativeentrop}. Since then, once generalizations of the quantum relative entropies were defined, the question of whether the data processing inequality holds or not (and for which parameters) generated several publications such as \cite{petz, petz1, summaryPaper, AdAlphaZrenyi, vershynina, vershynina1, Beigi13, sandwichedRecovery, sandwichConditions, jencova1,Jenvcova2,umegakiEntropy,CFLconjecture, hiaiconcavity, carlen2016some, alphaRecovery}. One quantum generalization of the quantum relative entropy is called the $\alpha - $R{\'e}nyi relative entropy, and it was proven by Petz in \cite{petz, petz1} and also by B{\'e}ny, Mosonyi, et.al in \cite{alphaRecovery}, that this entropy measure indeed satisfies the data processing inequality. More information about this relative entropy is found in section \ref{alphaRRE}. Later, a different generalization was introduced by M{\"u}ller-Lennert, Dupuis, Szehr, Fehr, and Tomamichel in \cite{sandwichDefinition}, and independently by Wilde, Winter, and Yang in \cite{wildeSandwiched}. This family is called the $\alpha -$Sandwiched R{\'e}nyi relative entropy. Under certain parameters of $\alpha$, this relative entropy also satisfies the data processing inequality. Lastly, Audenaert and Datta, in their paper \cite{AdAlphaZrenyi}, introduced a two parameter family of relative entropies that generalizes all entropy functions stated thus far. This family of entropies, the $\alpha-z$ R{\'e}nyi relative entropy, was completely characterized in terms of its two parameters, $\alpha$ and $z$, as to when it satisfies the data processing inequality and when it does not. Reference \cite{summaryPaper} gives a nice intuitive summary of the contributions to the data processing inequality, and \cite{CFLconjecture} finishes it up with the final characterizations of the parameters. 

It is known that Lindblad and Uhlmann  in \cite{lindblad1974expectations} and \cite{uhlmann1973} respectively proved that satisfying the data processing inequality is equivalent to proving convexity or concavity of certain trace functionals within the definitions of the relative entropy functions. This is a crucial ingredient in most of results on data processing. In fact, working with these trace functionals is just as important when answering the questions of whether it is possible to saturate these inequalities.

Our interest lies in the question of saturating the DPI. i.e., when is the relative entropy preserved when states pass through a noisy quantum channel. For some of the relative entropies above, the answer to this question is in terms of recoverability of states. Recoverability exists for the quantum relative entropy, the R{\'e}nyi relative entropy, and the $\alpha-$Sandwiched R{\'e}nyi relative entropy. This work contributes to the question of recoverability in terms of the $\alpha-z$ R{\'e}nyi relative entropy. 

After this paper first appeared on the arxiv, another paper \cite{zhang2020equality} appeared asking the same question: when is the DPI saturated for the $\alpha-z$ R{\'e}nyi relative entropy? In \cite{zhang2020equality}, the authors answer this question in a different format than it is presented here. More work is needed to compare the results.

The main result of this paper says:

\begin{corollary-non}[\ref{partial case1}]

Let $\rho\in\cD(\cH)$, $\sigma\in\mathcal{Q}(\mathcal{H})$, and $\Lambda:\cB(\cH)\to\cB(\cK)$ be a quantum channel. For any $1<\alpha\leq 2$ and $\frac{\alpha}{2}\leq z\leq\alpha$,  whenever saturation of the DPI holds, i.e., $D_{\alpha,z}(\rho||\sigma)=D_{\alpha,z}(\Lambda(\rho)||\Lambda(\sigma))$,   then the states satisfy  
\begin{equation}
\sigma^{\frac{1-z}{2z}}(\sigma^{\frac{1-\alpha}{2z}}\rho^{\frac{\alpha}{z}}\sigma^{\frac{1-\alpha}{2z}})^{z-1}\sigma^{\frac{1-z}{2z}}=\Lambda^*\left(\Lambda(\sigma)^{\frac{1-z}{2z}}(\Lambda(\sigma)^{\frac{1-\alpha}{2z}}\rho^{\frac{\alpha}{z}}\Lambda(\sigma)^{\frac{1-\alpha}{2z}})^{z-1}\Lambda(\sigma)^{\frac{1-z}{2z}}\right).
\end{equation}
\end{corollary-non}

The paper is arranged as follows: section $2$, discusses the definitions and notations used throughout this paper. In addition to this, some known results and properties about the different quantum relative entropies of interest are mentioned. In section $3$, the technical results using tools in complex analysis or results about convex/concave trace functionals are discussed. Section $4$ is dedicated to the main result in the context of partial traces followed by the more general consequences. Finally, section $5$ concludes with a brief discussion on closing remarks.

%%%%%%%%%%%%%%% Preliminaries %%%%%%%%%%
\section{Notations and Definitions}\label{sec: notations}

Throughout this paper, only finite-dimensional Hilbert spaces are considered. When $\cH_{AB}$ is written, it is understood to mean a tensor product of Hilbert spaces $\cH_A$ and $\cH_B$. For a Hilbert space $\cH$, let $\cB(\cH)$ denote the set of bounded linear operators on $\cH$. The set of all positive operators is denoted by 
$$\cP(\cH):=\{A\in\cB(\cH) : A> 0\},$$
and the space of all \emph{density operators} is defined to be 
$$\cD(\cH):=\{\rho\in\cP(\cH) : \mathrm{Tr}(\rho)=1\}.$$
Recall that if an operator $\rho$ is positive $(\rho>0$), then it is automatically self adjoint, i.e., $\rho^*=\rho$. Given any linear operator $\cL:\cB(\cH)\to \cB(\cK)$, where $\cH$ and $\cK$ are Hilbert spaces, the adjoint operator $\cL^*:\cB(\cK)\to \cB(\cH)$ is the unique operator satisfying 
$$\langle\cL(X),Y\rangle_{\cB(\cK)}=\langle X,\cL^*(Y)\rangle_{\cB(\cH)},$$
and the inner product here is the Hilbert-Schmidt inner product defined as 
$$\langle X,Y\rangle_{HS}=\mathrm{Tr}(X^*Y).$$
In general, the Schatten $p-$ norm is defined as 
$$\|X\|_p:=\left(\mathrm{Tr}\left[(X^*X)^{\frac{p}{2}}\right]\right)^{\frac{1}{p}},$$
for $p\in[1,\infty)$. Note that this norm satisfies the H{\"o}lder inequality, sub-multiplicativity, and monotonicity in $p$. A linear operator $\cL:\cB(\cH)\to \cB(\cK)$ is said to be \emph{$n-$positive} if
$$\mathbbm{1}_n\otimes\cL:\cB(\bC^n)\otimes\cB(\cH)\to \cB(\bC^n)\otimes\cB(\cK)$$
is a positive operator, where $\mathbbm{1}_n$ is the identity operator on $\cB(\bC^n)$. 

\begin{definition}
If $\mathbbm{1}_n\otimes\cL$ is positive for all $n\in\bN$, then $\cL$ is called a \emph{completely positive map}. 
\end{definition}

\begin{definition}
Completely positive maps that also preserve the trace of operators are called \emph{quantum channels}. i.e., $\mathrm{Tr}(\cL(\rho))$=$\mathrm{Tr}(\rho)$. 
\end{definition}

\subsection{Quantum Relative Entropies}
%%%%%%%%%%%%%%%%%%%%%%%%%%%%%%%%%%%%%%%%%%%%%%%%%%%%%
%%%%%%%%%% Umegaki Relative Entropy %%%%%%%%%%%%%%%%%%%%%%%%%%
%%%%%%%%%%%%%%%%%%%%%%%%%%%%%%%%%%%%%%%%%%%%%%%%%%%%%
\subsubsection{Umegaki Relative Entropy}\label{relativeentrop}
In quantum information theory, the information shared between states is regularly studied through the understanding of quantum entropies. Umegaki relative entropy, also know as the quantum relative entropy, in \cite{umegakiEntropy} is defined as $$D(\rho||\sigma):=\mathrm{Tr}(\rho\log\rho-\rho\log\sigma),$$
where $\rho\in\cD(\cH)$ and $\sigma\in\cD(\cH)$, provided that $\mathrm{supp}(\rho)\subseteq\mathrm{supp}(\sigma)$. Otherwise, the relative entropy between $\rho$ and $\sigma$ is said to be $\infty$. Reference \cite{summaryPaper} provides an extensive review of the formulation of this relative entropy, applications, and some of its properties. In this setting, \cite{umegakiEntropy} introduces and explains the DPI. That is, for any quantum channel $\Lambda$, the following inequality holds
    $$D(\rho||\sigma)\geq D(\Lambda(\rho)||\Lambda(\sigma)).$$
This inequality is interpreted as an increased difficulty in distinguishing states from one another after the states pass through a noisy quantum channel. 

\begin{definition}
    If there is a quantum channel $\Psi$ such that $\Psi$ recovers states $\rho$ and $\sigma$, i.e., 
    \begin{equation}(\Psi\circ\Lambda)\rho=\rho \mbox{ and } (\Psi\circ\Lambda)\sigma=\sigma,
\end{equation}
then we say that $\Lambda$ is \emph{sufficient} for states $\rho$ and $\sigma$. When this happens, the quantum channel $\Psi$ is called a \emph{recovery map}. 
\end{definition}

\noindent Saturation of the DPI was originally proven by Petz in the context of von Neumann algebras in \cite{petz, petz1}. The result states that 2 states $\rho$ and $\sigma$ saturate the DPI for a quantum channel $\Lambda$ if and only if the quantum channel is sufficient for these states. The recovery map $\Psi_{\sigma,\Lambda}$, known as the Petz recovery map has an explicit form of
    $$\Psi_{\sigma,\Lambda}(\cdot)=\sigma^{\frac{1}{2}}\Lambda^*\left(\Lambda(\sigma)^{-\frac{1}{2}}\cdot\Lambda(\sigma)^{-\frac{1}{2}}\right)\sigma^{\frac{1}{2}}.$$
The map $\Psi$ is indexed by $\sigma$ and $\Lambda$ to indicate that the recovery map depends on this state and this quantum channel. In a different context, you can saturate the DPI with the use of an error term, which was done in 
     
     \begin{theorem}\cite[Corollary 1.7]{vershynina1}
        For $\rho$ and $\sigma$ density operators and $\mathcal{N}$ a partial trace, the following inequality holds

        $$D(\rho||\sigma)-D(\mathcal{N}(\rho)||\mathcal{N}(\sigma))\geq (\frac{\pi}{8})^4\|\rho^{-1}\|^{-2}\|\mathcal{R}_\rho(\mathcal{N}(\sigma))-\sigma\|_1^4, $$ where $\|\cdot\|_1$ is the Schatten 1-norm, $\|X\|_1:=\mathrm{Tr}|X|=\mathrm{Tr}(X^\ast X)^{\frac{1}{2}}$, for an operator $X$ and $\mathcal{R}_\rho$ is a Petz recovery map.  
    \end{theorem}

%%%%%%%%%%%%%%%%%%%%%%%%%%%%%%%%%%%%%%%%%%%%%%%%%%%%%
%%%%%%%%%% Renyi Relative Entropy %%%%%%%%%%%%%%%%%%%%
%%%%%%%%%%%%%%%%%%%%%%%%%%%%%%%%%%%%%%%%%%%%%%%%%%%%%
 \subsubsection{\texorpdfstring{$\alpha-$}RR{\'e}nyi Relative Entropy (\texorpdfstring{$\alpha-$}RRRE)}\label{alphaRRE}
 One of the first generalizations of the Umegaki relative entropy is defined for $\alpha\in(-\infty,1)\cup(1,\infty)$. The $\alpha-$R{\'e}nyi Relative Entropy is expressed as $$D_\alpha(\rho||\sigma):=\frac{1}{\alpha-1}\log \mathrm{Tr}\left(\rho^\alpha\sigma^{1-\alpha}\right),$$
 provided $\mathrm{supp}(\rho)\subseteq \mathrm{supp}(\sigma)$. 
 For $\alpha\in[0,1)\cup(1,2]$, Theorem 5.1 of \cite{alphaRecovery} proves equality of the DPI if and only if there exists a recovery map that recovers both states $\rho$ and $\sigma$ perfectly well. Furthermore, \cite{alphaRecovery} also provides the algebraic necessary and sufficient conditions for the $\alpha$-R{\'e}nyi relative entropy as well. That is for all $\alpha\in[0,1)\cup(1,2]$, saturation of the DPI is satisfied if and only if 
$$\Lambda^*(\Lambda(\sigma)^{-z}\Lambda(\rho)^{-z})= \sigma^{-z}\rho^z,$$
for all $z\in\bC$. \footnote{It is understood that when $\alpha=1$ or the limit as $\alpha$ approaches $1$, this is the Umegaki relative entropy case.} These proofs are in a more general class of quantum functionals called quantum f-divergence, which is actually a class of quantum quasi-entropies. In fact in \cite{vershynina}, similar claims are made however in a different context and with the use of an error bound. Once the bound is proven, necessity and sufficiency follow very easily. Their result is

\begin{theorem}\cite[Theorem 6.1]{vershynina}
For any $\alpha\in(0,1)$, under explicit assumptions defined in the paper for $\rho$ and $\sigma$, the following inequality holds
$$
    D_\alpha(\rho||\sigma)-D_\alpha(\Lambda(\rho)||\Lambda(\sigma))\geq
$$
$$\frac{1}{1-\alpha}\log\left(1+K\|\Lambda(\sigma)^{\frac{1}{2}}\Lambda(\rho)^{-\frac{1}{2}}\rho^{\frac{1}{2}}-\sigma^{\frac{1}{2}}\|_2^{6-2\alpha}\right),
$$
where $K$ is a constant calculated in their paper.
\end{theorem}

%%%%%%%%%%%%%%%%%%%%%%%%%%%%%%%%%%%%%%%%%%%%%%%%%%%%%
%%%%%%%%%% Sandwiched Renyi Entropy %%%%%%%%%%%%%%%%%
%%%%%%%%%%%%%%%%%%%%%%%%%%%%%%%%%%%%%%%%%%%%%%%%%%%%%
 \subsubsection{\texorpdfstring{$\alpha-$}SSandwiched R{\'e}nyi Divergence (\texorpdfstring{$\alpha-$}SSRD)}\label{sandwichedRRE} This section describes another way to generalize the Umegaki relative entropy and mentions some known results as well.  For $\alpha\in(-\infty,1)\cup(1,\infty)$, \cite{sandwichDefinition} and \cite{wildeSandwiched} introduce a new family of R{\'e}nyi relative entropies called \emph{Sandwiched R{\'e}nyi Relative Entropies} or \emph{Sandwiched R{\'e}nyi Divergence}. They are defined as  $$\tilde{D}_\alpha(\rho||\sigma):=\frac{1}{\alpha-1}\log\left(\mathrm{Tr}\left[\sigma^{\frac{1-\alpha}{2\alpha}}\rho\sigma^{\frac{1-\alpha}{2\alpha}}\right]^\alpha\right),$$ provided that $\mathrm{supp}(\rho)\subseteq \mathrm{supp}(\sigma)$. The DPI for all $\alpha\in[\frac{1}{2},1)\cup(1,\infty)$ was proven in \cite{Beigi13}. That is 
     $$\tilde{D}_\alpha(\rho||\sigma)\geq \tilde{D}_\alpha(\Lambda(\rho)||\Lambda(\sigma)).$$ 
In \cite{sandwichConditions}, it was shown that the equality of the DPI is satisfied for all $\alpha>\frac{1}{2}$, if and only if states $\rho$ and $\sigma$ have the following algebraic form:         $$\sigma^{\frac{1-\alpha}{2\alpha}}(\sigma^{\frac{1-\alpha}{2\alpha}}\rho\sigma^{\frac{1-\alpha}{2\alpha}})^{\alpha-1}\sigma^{\frac{1-\alpha}{2\alpha}}=\Lambda^*\left(\Lambda(\sigma)^{\frac{1-\alpha}{2\alpha}}(\Lambda(\sigma)^{\frac{1-\alpha}{2\alpha}}\rho\Lambda(\sigma)^{\frac{1-\alpha}{2\alpha}})^{\alpha-1}\Lambda(\sigma)^{\frac{1-\alpha}{2\alpha}}\right).$$ 
For $\alpha>1$, the equivalence between saturating the DPI and the sufficiency property is proven in \cite{sandwichedRecovery}. The techniques use non-commutative interpolated $L_p$ spaces for von Neumann Algebras. The same author proved the equivalence between saturating the DPI and the sufficiency for $\alpha\in(\frac{1}{2},1)$ in \cite{Jenvcova2} using different norms.

 %%%%%%%%%%%%%%%%%%%%%%%%%%%%%%%%%%%%%%%%%%%%%%%%%%
 %%%%%%%%%%%%%%%%% ALPHA-Z %%%%%%%%%%%%%%%%%%%%%%
 %%%%%%%%%%%%%%%%%%%%%%%%%%%%%%%%%%%%%%%%%%%%%%%%
\subsubsection{\texorpdfstring{$\alpha-z $ }RR{\'e}nyi Relative Entropy (\texorpdfstring{$\alpha-z$ }RRRE)}\label{alphaZ}
Here is another generalization of relative entropy that combines both $\alpha$-RRE and $\alpha$-SRD.
Let $\rho,\sigma\in\cD(\cH)$ with $\alpha\in\bR\setminus\{1\}$ and $z>0$. The $\alpha-z$ R{\'e}nyi relative entropy was introduced in \cite{AdAlphaZrenyi} and is defined as
\begin{equation}
D_{\alpha,z}(\rho||\sigma):=\frac{1}{\alpha-1}\log\left(\mathrm{Tr}\left[\left(\sigma^{\frac{1-\alpha}{2z}}\rho^{\frac{\alpha}{z}}\sigma^{\frac{1-\alpha}{2z}}\right)^z\right]\right),
\end{equation}
provided the $\mathrm{supp}(\rho)\subseteq\mathrm{supp}({\sigma})$. Otherwise, the $\alpha-z$ R{\'e}nyi relative entropy is said to be $+\infty$. When $z=1$, the $\alpha-z$ RRE reduces to the $\alpha$-RRE. When $z=\alpha$, the $\alpha-z$ RRE reduces to the $\alpha$-SRD. When $z=1$ and $\alpha\to1$ or when $z=\alpha$ and $\alpha\to1$, the $\alpha-z$ RRE reduces to the Umegaki relative entropy.

Many of the interesting properties of the $\alpha-z$ entropies are explained and introduced in previous works on this family of entropies such as \cite{AdAlphaZrenyi, summaryPaper}. Only the properties that are explicitly used in this paper are listed here. 
\begin{enumerate}
    \item \label{invariance}
    \emph{Invariance:} The $\alpha-z$ R{\'e}nyi entropies are invariant under unitaries. That is for any unitary $U$, 
    $$D_{\alpha,z}(U\rho U^*||U\sigma U^*)= D_{\alpha,z}(\rho||\sigma).$$ 
    This is because for any unitary $U$ and for any operator $A$, the eigenvalues of $A$ and $UAU^*$ are the same.
    \item \label{tensor} 
    \emph{Tensor Property:} For any $\rho,\sigma,\tau\in\cD(\cH)$, 
    $$D_{\alpha,z}(\rho\otimes\tau||\sigma\otimes\tau)=D_{\alpha,z}(\rho||\sigma). $$
    This is due to the fact that the trace of a tensor product between two states is the product of trace of states.
\end{enumerate}

\begin{remark}
In this paper, it is always assumed that the operators, $\rho$ and $\sigma$, are invertible and that $\mathrm{supp}(\rho)\subseteq\mathrm{supp}(\sigma)$.
\end{remark}
The conjecture for which parameters of $\alpha$ and $z$ the DPI holds is outlined in \cite{summaryPaper} and finally concluded in \cite{CFLconjecture}. This is summarized in the next theorem. 
\begin{theorem}\cite[Theorem 1.2]{CFLconjecture}\label{alphazDPI}
The $\alpha-z$ R{\'e}nyi relative entropy is monotone under completely positive trace preserving maps (quantum channels) on $\cD(\cH)$ for all $\cH$ if and only if one of the following holds
\begin{enumerate}
    \item $0<\alpha<1$ and $z\geq\max{\{\alpha,1-\alpha\}}$;
    \item $1<\alpha\leq2$ and $\frac{\alpha}{2}\leq z\leq \alpha$; 
    \item $2\leq \alpha<\infty$ and $\alpha-1\leq z\leq\alpha$.
\end{enumerate}
\end{theorem}

\noindent One way to prove this is through the relationship between the DPI and joint convexity/concavity of the trace functional defined by the map 
\begin{equation}\label{functional}
(A,B)\mapsto\mathrm{Tr}(B^{\frac{q}{2}}K^*A^pKB^{\frac{q}{2}})^s,
\end{equation}
where $A$ and $B$ are positive operators on $\cH$, $K$ is any operator in $\cB(\cH)$, $p,q>0$, and $s\geq\frac{1}{p+q}$. Here, the case of interest is whenever $K$ is the identity operator, $p=\frac{\alpha}{z}$, $q=\frac{1-\alpha}{z}$, $s=\frac{1}{p+q}$, $A=\rho$, and $B=\sigma$. Then the trace functional is defined as
$$\Psi_{\alpha,z}(\rho||\sigma):=\mathrm{Tr}\left[\left(\sigma^{\frac{1-\alpha}{2z}}\rho^{\frac{\alpha}{z}}\sigma^{\frac{1-\alpha}{2z}}\right)^z\right]. $$

The next theorem describes the relationship between the DPI and joint convexity and joint concavity of the trace functional.

\begin{theorem}\cite[Proposition 7]{summaryPaper}\label{dpi1}
Let $\alpha,z>0$ with $\alpha\neq1$. Then $D_{\alpha,z}$ is monotone under quantum channels on $\cP(\cH)$  (for any finite dimensional $\cH$) if and only if one of the following holds:
\begin{enumerate}
    \item $\alpha<1$ and $\Psi_{\alpha,z}(\rho||\sigma)$ is jointly concave.
    \item $\alpha>1$ and $\Psi_{\alpha,z}(\rho||\sigma)$ is jointly convex.
\end{enumerate}
\end{theorem}

\noindent Together, Theorem \ref{alphazDPI} and Theorem \ref{dpi1} give the complete picture for data processing inequality of the $\alpha-z$ R{\'e}nyi Relative Entropy.

$ $ 

The next two sections are the technical components used to prove the main result:

%%%%%%%%%%%%%%% Pick Functions %%%%%%%%%%
\section{Preliminaries}\label{sec: prelims}

\subsection{Tools from Complex Analysis}
Let $\mathrm{Spec}(X)$ denote the set of eigenvalues for operator $X$. If $\Omega$ is an open subset of $\mathbbm{C}$ such that $\mathrm{Spec}(X)\subseteq\Omega\subseteq\mathbbm{C}$, then the analytic functional calculus is used to conclude that $F(X)$ is well defined, for any analytic function $F$. Define 
$$\mathbbm{C}^+ := \{z\in\mathbbm{C} \mbox{ s.t. } \mathrm{Im}(z)>0\}.$$ Note that $\mathbbm{C}^+$ is an open subset of $\mathbbm{C}$. For any operator $X$, define
$$ \mathrm{Re}(X) := \frac{X + X^*}{2} \text{ and } \mathrm{Im}(X) := \frac{X - X^*}{2i}.$$
Then let 
$$ I_n^+:= \{X \in\mathbbm{M}_n(\mathbbm{C}) \mbox{ s.t. } \mathrm{Im}(X)>0\} $$ 
and 
$$I_n^-:= \{X \in\mathbbm{M}_n(\mathbbm{C}) \mbox{ s.t. } \mathrm{Im}(X)<0\},$$ 
where $n<\infty$ denotes the dimension. For $0<p\leq1$, denote 
$$ \Gamma_{p\pi} := \{re^{i\theta} : r>0 \text{ and } 0<\theta<p\pi\} $$
and 
$$\Gamma_{-p\pi} := \{re^{i\theta} : r>0 \text{ and } 0>\theta>-p\pi\}.$$  Note that when $p=1$, $\Gamma_{p\pi}=\mathbbm{C}^+$. Let us recall a few known facts or results.

\begin{lemma}\cite[Lemma 1.1]{hiai2001concavity}\label{hiailemma1}
If $X\in I_n^+$,  then $X$ is invertible and $\mathrm{Spec}(X)\subset\mathbbm{C}^+$. For $X$, an invertible $n\times n$ matrix,  $X\in I_n^+$ if and only if $X^{-1}\in I_n^-$. 
\end{lemma}

\begin{lemma}\cite[Lemma 1.2]{hiai2001concavity}\label{hiailemma2}
Let $0<p\leq1$. If $X\in I^+_n$, then so is $X^p$ and $e^{-ip\pi}X^p\in I_n^-$. If $X\in I_n^-$, then $X^p\in I_n^-$ and $e^{ip\pi}X^p\in I_n^+$. 
\end{lemma}

\begin{lemma}\label{conjugation}
Any pair of operators $A$ and $B$ have the following properties:
$$\mathrm{Re}(ABA^*)=A(\mathrm{Re}B)A^* \mbox{ }\mathrm{and}\mbox{ }  \mathrm{Im}(ABA^*)=A(\mathrm{Im}B)A^*.$$
\end{lemma}

\begin{proof}
Observe that 
\begin{align}
    \mathrm{Im}(ABA^*) & = \frac{(ABA^*) - (ABA^*)^*}{2i} \\
    & = \frac{(ABA^*)-(AB^*A^*)}{2i} \\
    & = \frac{A(B-B^*)A^*}{2i} \\
    & = A(\mathrm{Im}B)A^*,
\end{align} and the real version is similar. 
\end{proof}

\begin{lemma}\label{pos.conj}
If $B\in I^+_n$, then so is $ABA^*$, for any non-zero operator $A$. 
\end{lemma}

\begin{proof}
Lemma \ref{conjugation}, implies that $\mathrm{Im}(ABA^*)=A(\mathrm{Im}B)A^*$. Thus for any $y$, 
\begin{align}
    \left\langle 
    \mathrm{Im}(ABA)^*y,y\right\rangle & = 
    \left\langle 
    A(\mathrm{Im}B)A^*y,y\right\rangle \\
    & = \left\langle (\mathrm{Im}B)A^*y,A^*y\right\rangle \\
    & = \left\langle \left[(\mathrm{Im}B)^{\frac{1}{2}}\right]A^*y,\left[(\mathrm{Im}B)^{\frac{1}{2}}\right]A^*y\right\rangle >0,
\end{align} where the last equality holds because $\mathrm{Im}(B)$ is positive providing a unique square root that is also positive. 
\end{proof}

\begin{remark}\label{remark: newset}
Define
\begin{equation}\label{eq: big.assumption.set}
  \mathcal{I}(\mathcal{H}):=  \left\{X \in\mathcal{P}(\mathcal{H}): (zX+H)^{\frac{p}{2}}B\,(zX+H)^{\frac{p}{2}}\in I^+_n \right\},
\end{equation}
for all $H$ hermitian, $z\in\mathbbm{C}^+$, $0<p\leq1$, and $B\in\mathcal{P}(\mathcal{H})$.
\end{remark}

\begin{proposition}\label{newPick}
Define $\Phi_Z(X) := Z^*XZ$, where $Z$ is an invertible operator. For any $0<p<1$, a positive operator $X$ belonging to any convex subset of $\mathcal{I}(\mathcal{H})$ is such that the map 
 
\begin{equation}
    X \mapsto \mathrm{Tr}\left[\left\{\Phi_Z\left(X^{\frac{p}{2}}\right)A \,\Phi_Z\left(X^{\frac{p}{2}}\right)\right\}^{\frac{1}{p}}\right]
\end{equation}
is concave for any $A\in\cP(\cH)$.
\end{proposition}

\begin{proof}[Proof of Proposition \ref{newPick}]
It suffices to show that for any $H$ hermitian,
\begin{equation}\label{2derivative}
    \frac{d^2}{dx^2}\left(\mathrm{Tr}\left[\left\{\Phi_Z\left((X+xH)^{\frac{p}{2}}\right)A \,\Phi_Z\left((X+xH)^{\frac{p}{2}}\right)\right\}^{\frac{1}{p}}\right]\right) \leq 0,
\end{equation}
for any small $x>0$. This is because if (\ref{2derivative}) holds, then
$$(X + xH)\mapsto \mathrm{Tr}\left[\left\{\Phi_Z\left((X+xH)^{\frac{p}{2}}\right)A \,\Phi_Z\left((X+xH)^{\frac{p}{2}}\right)\right\}^{\frac{1}{p}}\right] $$ is concave, and hence the proposition is proved by taking $x\to 0$. 

Observe that for any $z\in\mathbbm{C}^+$, it follows that $zX+H\in I_n^+$. So $\left(zX+H\right)^{\frac{p}{2}}$ is well defined, by the analytic functional calculus. Moreover, by Lemma \ref{hiailemma2}, observe that $\left(zX+H\right)^{\frac{p}{2}}\in I^+_n$. By linearity of $\Phi_Z$ and by Lemma \ref{pos.conj}, $\Phi_Z\left((zX+H)^{\frac{p}{2}}\right)\in I_n^+$. Define
\begin{equation}\label{bigF}
    F(z) := \Phi_Z\left((zX+H)^{\frac{p}{2}}\right)A \,\Phi_Z\left((zX+H)^{\frac{p}{2}}\right). 
\end{equation} This function is analytic in $\mathbbm{C}^+$ because a product of analytic functions is analytic. 
%Since $A\in I^+_n$, by Lemma \ref{pos.conj} and by the fact that $\Phi$ is self adjoint, we see that $F(z)\in I^+_n$. Thus by Lemma \ref{hiailemma1}, $\mathrm{Spec}(F(z)) \in \mathbbm{C}^+$, which is an open subset of $\mathbbm{C}$.

To prove equation (\ref{2derivative}), we prove the following steps:
\begin{enumerate}
    \item Show $\mathrm{Spec}(F(z))$ is contained in some open subset of $\mathbbm{C}$ so that $(F(z))^{\frac{1}{p}}$ is well defined.
    \item Extend $F(z)$ onto the real line $(R,\infty)$.
    \item Express $\mathrm{Tr}\left[(F(z))^{\frac{1}{p}}\right]$ as a Pick function to admit an integral representation.
    \item Show concavity of the above expression through its integral representation.
\end{enumerate}

\underline{Step 1:} We show $\mathrm{Spec}(F(z)) \subset \Gamma_{p\pi}$. As in \cite{hiai2001concavity}, it suffices to prove the following 3 properties:
\begin{enumerate}
    \item $\mathrm{Spec}(F(z))\subset\Gamma_{p\pi}$, when $z=re^{i\theta}$ with fixed $0<\theta<\pi$ and sufficiently large $r>0$.
    \item $\mathrm{Spec}(F(z))\cap[0,\infty)=\emptyset$ for all $z\in\mathbbm{C}^+$.
    \item $\mathrm{Spec}(F(z))\cap \{re^{ip\pi}\}=\emptyset$ for all $z\in\mathbbm{C}^+$.
\end{enumerate}
These statements are sufficient because if, for the sake of contradiction, $\mathrm{Spec}(F(z)) \not\subset \Gamma_{p\pi}$ for some $z_o=r_oe^{i\theta_o}\in\mathbbm{C}^+$ and the 3 properties are satisfied, then by continuity of eigenvalues of $F(z)$ and by property $(1)$, $\mathrm{Spec}(F(z))\cup \partial\Gamma_{p\pi}\neq\emptyset$, for some \newline $z\in\{re^{i\theta_o} : r>r_o\}$, which implies either property (2) or property (3) is violated. 

To prove property $(1)$, by linearity of $z$ in $F(z)$,
$$F(z) = z^p\left[\Phi_Z\left((X+z^{-1}H)^{\frac{p}{2}}\right) A \, \Phi_Z\left((X+z^{-1}H)^{\frac{p}{2}}\right)\right]. $$ 
Whenever $z=re^{i\theta_o}$ with a fixed $0<\theta_o<\pi$, note that $$\mathrm{Spec}\left[\Phi_Z\left((X+z^{-1}H)^{\frac{p}{2}}\right) A \, \Phi_Z\left((X+z^{-1}H)^{\frac{p}{2}}\right)\right]$$ converges to $$\mathrm{Spec}\left[\Phi_Z\left((X)^{\frac{p}{2}}\right) A \, \Phi_Z\left((X)^{\frac{p}{2}}\right)\right]\subset(0,\infty), \text{ as }r\to\infty.$$ 

To prove property $(2)$, for any $0\leq r<\infty$, \begin{equation}\label{eq: F(z)-id}
    F(z)-r\mathbbm{1}_n=\Phi_Z\left((zX+H)^{\frac{p}{2}}\right) (A - r\Phi_Z((zX+H)^{p/2})^{-2})\, \Phi_Z\left((zX+H)^{\frac{p}{2}}\right).
\end{equation} By the assumption of $X$ in any convex subset of $\mathcal{I}(\cH)$ and by Lemma \ref{hiailemma1}, $$(zX+H)^{-\frac{p}{2}}(Z^{-1})^*(Z^{-1})(zX+H)^{-\frac{p}{2}}\in I^-_n,$$ and by Lemma \ref{conjugation},
$$\Phi_{(Z^{-1})^*}\left[(zX+H)^{-\frac{p}{2}}(Z^{-1})^*(Z^{-1})(zX+H)^{-\frac{p}{2}}\right]\in I^-_n .$$ Note that $\left[\Phi_Z\left((zX+H)^{p/2}\right)\right]^{-2} = \Phi_{(Z^{-1})^*}\left[(zX+H)^{-\frac{p}{2}}(Z^{-1})^*(Z^{-1})(zX+H)^{-\frac{p}{2}}\right] \in I^-_n $.
This implies that $A-r\left[\Phi_Z\left((zX+H)^{p/2}\right)\right]^{-2}\in I^+_n$,
and thus $F(z)-r\mathbbm{1}_n$ is invertible by Lemma \ref{hiailemma1}. If there exists $l\in \mathrm{Spec}(F(z))\cap[0,\infty)$, then  $F(z)-l\mathbbm{1}_n=0$, which would contradict invertability for all $r\in[0, \infty)$. Hence property $(2)$ is satisfied. Proving property $(3)$ is very similar to proving the second property using the assumption from  $X$ in any convex subset of $\mathcal{I}(\cH)$ as well as  Lemma \ref{hiailemma2}. With this and using the analytic functional calculus,  $(F(z))^{\frac{1}{p}}$ is well defined. %So  $\mathrm{Tr}((F(z))^{\frac{1}{p}})\in\mathbbm{C}^+$ holds, by Lemma \ref{hiailemma1}. 

\underline{Step 2:} For every $z\in\mathbbm{C}^+$ such that $|z|>R$, define an analytic function
\begin{equation}\label{new.bigF}
    \hat{F}(z):= z^p\Phi_Z\left((X+z^{-1}H)^{\frac{p}{2}}\right) A\,\Phi_Z\left((X+z^{-1}H)^{\frac{p}{2}}\right).
\end{equation}
Then continuously extending this function to the real line, for every $x\in(R,\infty)$ gives $\hat{F}(x)=F(x)$. Hence for every $z\in\mathbbm{C}^+$ such that $|z|>R$, write 
$$F(z)= z^p\Phi\left((X+z^{-1}H)^{\frac{p}{2}}\right)A \,\Phi\left((X+z^{-1}H)^{\frac{p}{2}}\right),$$ and
$$(F(z))^{\frac{1}{p}}= z\left\{\Phi_Z\left((X+z^{-1}H)^{\frac{p}{2}}\right)A \,\Phi_Z\left((X+z^{-1}H)^{\frac{p}{2}}\right)\right\}^{\frac{1}{p}}.$$

\underline{Step 3:} Given that $\mathrm{Tr}\left[(F(z))^{\frac{1}{p}}\right]\in\mathbbm{C}^+$ for every $z\in\mathbbm{C}^+$,  and $\mathrm{Tr}\left[(F(x))^{\frac{1}{p}}\right]\in\mathbbm{R}$ for every $x\in(R,\infty)$, by the Schwarz Reflection Principle, $(F(z))^{\frac{1}{p}}$ can be extended to the lower half plane, that is the set of complex number with negative imaginary parts. And thus, $\phi$ is a Pick function on $\mathbbm{C}\setminus(-\infty, R)$, where $\phi(x) = \mathrm{Tr}\left[(F(x))^{\frac{1}{p}}\right]$ for all $x\in(R, \infty)$. Then
\begin{equation}
    x\phi(x^{-1}) = \mathrm{Tr}\left\{\Phi_Z\left((X+xH)^{\frac{p}{2}}\right)A \,\Phi_Z\left((X+xH)^{\frac{p}{2}}\right)\right\}^{\frac{1}{p}},
\end{equation}
for every $x\in(0, R^{-1})$. By theory of Pick functions, see \cite{bhatia2013matrix}, every Pick function $\phi$ admits an integral representation. 
\begin{equation}
    \phi(z) = a + bz + \int_{-\infty}^\infty \frac{1+tz}{t-z}d\nu(t),
\end{equation}
where $a\in\mathbbm{R}$, $b\geq0$, and $\nu$ is a finite measure on $\mathbbm{R}$. The measure $\nu$ is supported in $(-\infty, R]$ because $\phi$ is analytically continued across $(R, \infty)$. 

\underline{Step 4:} For all $x\in(0,R^{-1})$, 
\begin{align}
x\phi(x^{-1}) & = x\left(a + \frac{b}{x} + \int_{-\infty}^\infty \frac{1+tx^{-1}}{t-x^{-1}}d\nu(t)\right) \\ 
& = ax + b + \int_{-\infty}^R \frac{x(x+t)}{tx-1}d\nu(t),
\end{align}
with 
\begin{align}
    \frac{d^2}{dx^2}\left(ax + b + \int_{-\infty}^R \frac{x(x+t)}{tx-1}d\nu(t)\right) & = \int_{-\infty}^R\frac{d^2}{dx^2}\left( \frac{x(x+t)}{tx-1} \right) d\nu(t) \\
    & = \int_{-\infty}^R \left(\frac{2(t^2+1)}{(xt-1)^3}\right) d\nu(t) <0,
\end{align}
for all $x\in(0,R^{-1})$ and all $t\in(-\infty, R)$.  
\end{proof}

\subsection{Convex and Concave Trace Functionals}
%%%%%%%%%%% f function %%%%%%%%%%%%

%%%%%%%%%%% f function %%%%%%%%%%%%

%%%%%%%%%%%%%%%%%%%%%%%%%%%%%%%%%%%%%

%%%%%%%%%%% convexity of f %%%%%%%%%%%%

To prove joint convexity of $f_{\alpha, z}(H,\rho, \sigma)$ from equation (\ref{f1}) below, the next few results are needed. Recall the equation from (\ref{functional})
\begin{equation}\label{functional1}
(A,B)\mapsto\mathrm{Tr}(B^{\frac{q}{2}}K^*A^pKB^{\frac{q}{2}})^s.
\end{equation}

\begin{theorem}\cite[Theorem 1.1]{CFLconjecture}\label{Reduced1}
Fix any invertible matrix $\cK$. Suppose that $p\geq q$ and $s>0$. 
\begin{enumerate}
    \item If $0\leq q\leq p\leq 1$ and $0<s\leq\frac{1}{p+q}$, then the map from (\ref{functional1}) is jointly concave.
    \item If $-1\leq q\leq p\leq 0$ and $s>0$, then the map from (\ref{functional1}) is jointly convex. 
    \item If $-1\leq q \leq 0$, $1\leq p\leq2$, $(p,q)\neq (1,-1)$ and $s\geq\frac{1}{p+q}$, then the map from (\ref{functional1}) is jointly convex.
\end{enumerate}
\end{theorem}

\begin{theorem}\cite[Theorem 3.3]{CFLconjecture}\label{zhangresult}
For $r_i>0$, $i\in\{0,1,2\}$ such that $\frac{1}{r_0}=\frac{1}{r_1}+\frac{1}{r_2}$, one has that for any invertible $X,Y\in\cB(\cH)$ that
\begin{equation}
    \mathrm{Tr}|XY|^{r_1} = \max\left\{ \frac{r_1}{r_0} \mathrm{Tr}|XZ|^{r_0}-\frac{r_1}{r_2}\mathrm{Tr}|Y^{-1}Z|^{r_2} : Z\in\cB(\cH) \mbox{ and invertible} \right\}.
\end{equation}

\end{theorem}

\begin{proposition}\cite[Proposition 5]{summaryPaper}\label{fixedB}
For a fixed operator $B$, the map on positive operators 
$$A\mapsto \mathrm{Tr}\left[(B^*A^pB)^{\frac{1}{p}}\right] $$ \begin{enumerate}
    \item is concave for $0\leq p\leq 1$, with $p\neq0$.
    \item is convex for $1\leq p\leq 2$, with $p\neq0$.
\end{enumerate}
\end{proposition}

\begin{proposition}\label{sum}
If $f:\cD(\cH)\times\cD(\cH)\times\cP(\cH)\mapsto[0,\infty)$ is defined as 
$$f(A,B,H) := g(A,H) + h(B,H), $$
where $g$ and $h$ are continuous, the functional $g$ is convex in $A$ and the functional $h$ is convex in $B$, then $f$ is jointly convex in $(A,B)$. Moreover, $\sup_{H>0}\{f(A,B,H)\}$ is jointly convex in $(A,B)$ whenever $f$ is.  
\end{proposition}

\begin{proof}
For all $i$ such that $0\leq\lambda_i\leq1$, with $\sum_i\lambda_i=1$, by convexity of $g$ and $h$, 
$$g\left(\sum_i\lambda_iA_i, H\right) \leq \sum_i\lambda_ig(A_i, H) \text{ and } h\left(\sum_i\lambda_iB_i, H\right) \leq \sum_i\lambda_ih(B_i, H). $$ 
Thus
\begin{align}
    f\left(\sum_i\lambda_iA_i,\sum_i\lambda_iB_i,H \right) & = g\left(\sum_i\lambda_iA_i, H\right) + h\left(\sum_i\lambda_iB_i, H\right) \\
    & \leq \sum_i\lambda_ig(A_i, H) + \sum_i\lambda_ih(B_i, H) \\
    & = \sum_i\lambda_i\left(g(A_i, H) + h(B_i, H)\right) \\
    & = \sum_i\lambda_if(A_i, B_i, H),
\end{align}
as desired. 
\end{proof}

%%%%%%%%%%% convexity of f %%%%%%%%%%%%

%%%%%%%%%%%%%%%%%%%%%%%%%%%%%%%%%%%%%

  Recall that $$\Psi_{\alpha,z}(\rho||\sigma):=\mathrm{Tr}\left[\left(\sigma^{\frac{1-\alpha}{2z}}\rho^{\frac{\alpha}{z}}\sigma^{\frac{1-\alpha}{2z}}\right)^z\right]. $$

\begin{lemma}\label{function2}
Let $\rho$ and $\sigma\in\cP(\cH)$, and assume that $\alpha>1$ and $z>1$. For any positive operator $H$, define
\begin{equation}\label{f1}
     f_{\alpha,z}(H,\rho,\sigma) :=      z\mathrm{Tr}(\sigma^{\frac{z-\alpha}{2z}}\rho^{\frac{\alpha}{z}}\sigma^{\frac{z-\alpha}{2z}}H) - (z-1)\mathrm{Tr}\left[(\sigma^{\frac{z-1}{2z}}H\sigma^{\frac{z-1}{2z}})^{\frac{z}{z-1}}\right].
\end{equation}
Then
\begin{equation}
    \Psi_{\alpha,z}(\rho||\sigma)=\sup_{H>0} f_{\alpha,z,H}(\rho,\sigma),
\end{equation}
where the supremum is achieved whenever $H = \sigma^{\frac{1-z}{2z}}(\sigma^{\frac{1-\alpha}{2z}}\rho^{\frac{\alpha}{z}}\sigma^{\frac{1-\alpha}{2z}})^{z-1}\sigma^{\frac{1-z}{2z}}$. 

\end{lemma}
\begin{proof}

For $X$ and $Y$ $\in\cP(\cH)$, it follows that $\mathrm{Tr}(XY) = \mathrm{Tr}\left(X^{\frac{1}{2}}YX^{\frac{1}{2}}\right)$, which is positive because $X^{\frac{1}{2}}YX^{\frac{1}{2}}$ is a positive operator. For a choice of  $1\leq p,q\leq\infty$, such that $1=\frac{1}{p}+\frac{1}{q}$, 
\begin{align}
   0 \leq  \mathrm{Tr}(XY) \label{positive}= & |\mathrm{Tr}(XY)|  \\
   \leq & \mathrm{Tr}|XY| \label{abs}\\
   \leq &\left(\mathrm{Tr}(X^p)\right)^{\frac{1}{p}}\left(\mathrm{Tr}(Y^q)\right)^{\frac{1}{q}} \label{line: holder} \\
    \leq & \frac{1}{p}\mathrm{Tr}(X^p) + \frac{1}{q}\mathrm{Tr}(Y^q)\label{line: young}
\end{align}
where (\ref{abs}) is standard for operators, (\ref{line: holder}) follows from Theorem 1.1 in \cite{manjegani2007holder} for positive operators $X$ and $Y$, and (\ref{line: young}) follows from the standard Young's inequality. Note from Theorem 1.1 in \cite{manjegani2007holder} that line (\ref{line: holder}) is saturated if and only if $X^p=Y^q$ which also implies equality of line (\ref{line: young}) as well. Moreover, $X=Y^{\frac{q}{p}}$ also implies that 
\begin{align}
    \mathrm{Tr}|XY| = & \mathrm{Tr}|Y^{\frac{q}{p}}Y| \\
    = & \mathrm{Tr}|Y^{\frac{q+p}{p}}| \\
    = & \mathrm{Tr}(Y^q) \\
    = & \mathrm{Tr}(Y^{\frac{q}{p}}Y) \\
    = & \mathrm{Tr}(XY),
\end{align}
so that (\ref{positive}) is also saturated whenever $X^p=Y^q$. Thus, 
\begin{equation}\label{inequality1}
    p\mathrm{Tr}(XY) - \frac{p}{q}\mathrm{Tr}(Y)^{q} \leq \mathrm{Tr}X^{p}.
\end{equation}
Take  positive operators 
$$ X=(\sigma^{\frac{1-\alpha}{2z}}\rho^{\frac{\alpha}{z}}\sigma^{\frac{1-\alpha}{2z}}) \mbox{ and } Y=(\sigma^{\frac{z-1}{2z}}H\sigma^{\frac{z-1}{2z}}), $$
where $H$ is some positive operator in $\cB(\cH)$. Then it follows that (\ref{inequality1}) becomes
\begin{equation}\label{newIneq1}
    p\mathrm{Tr}(\sigma^{\frac{1-\alpha}{2z}}\rho^{\frac{\alpha}{z}}\sigma^{\frac{1-\alpha}{2z}}\sigma^{\frac{z-1}{2z}}H\sigma^{\frac{z-1}{2z}}) - \frac{p}{q}\mathrm{Tr}\left[(\sigma^{\frac{z-1}{2z}}H\sigma^{\frac{z-1}{2z}})^{q}\right] \leq \mathrm{Tr}\left[(\sigma^{\frac{1-\alpha}{2z}}\rho^{\frac{\alpha}{z}}\sigma^{\frac{1-\alpha}{2z}})^{p}\right],
\end{equation}
where equality happens if and only if
$$ \sigma^{\frac{1-z}{2z}}\left(\sigma^{\frac{1-\alpha}{2z}}\rho^{\frac{\alpha}{z}}\sigma^{\frac{1-\alpha}{2z}}\right)^{\frac{p}{q}}\sigma^{\frac{1-z}{2z}}=H.$$

By construction, $H$ is unique, hence the left hand side of (\ref{inequality1}) becomes 
\begin{align}
     &\hspace*{-1cm} p\mathrm{Tr}\left(\sigma^{\frac{1-\alpha}{2z}}\rho^{\frac{\alpha}{z}}\sigma^{\frac{1-\alpha}{2z}}\sigma^{\frac{z-1}{2z}}\sigma^{\frac{1-z}{2z}}\left(\sigma^{\frac{1-\alpha}{2z}}\rho^{\frac{\alpha}{z}}\sigma^{\frac{1-\alpha}{2z}}\right)^{\frac{p}{q}}\sigma^{\frac{1-z}{2z}}\sigma^{\frac{z-1}{2z}}\right) - \\
     & \frac{p}{q}\mathrm{Tr}\left[\left(\sigma^{\frac{z-1}{2z}}\sigma^{\frac{1-z}{2z}}\left(\sigma^{\frac{1-\alpha}{2z}}\rho^{\frac{\alpha}{z}}\sigma^{\frac{1-\alpha}{2z}}\right)^{\frac{p}{q}}\sigma^{\frac{1-z}{2z}}\sigma^{\frac{z-1}{2z}}\right)^{q}\right] \\
     = & p\mathrm{Tr}\left(\sigma^{\frac{1-\alpha}{2z}}\rho^{\frac{\alpha}{z}}\sigma^{\frac{1-\alpha}{2z}}\left(\sigma^{\frac{1-\alpha}{2z}}\rho^{\frac{\alpha}{z}}\sigma^{\frac{1-\alpha}{2z}}\right)^{\frac{p}{q}}\right) - \frac{p}{q}\mathrm{Tr}\left[\left(\left(\sigma^{\frac{1-\alpha}{2z}}\rho^{\frac{\alpha}{z}}\sigma^{\frac{1-\alpha}{2z}}\right)^{\frac{p}{q}}\right)^{q}\right] \\
     = & p\mathrm{Tr}\left(\left(\sigma^{\frac{1-\alpha}{2z}}\rho^{\frac{\alpha}{z}}\sigma^{\frac{1-\alpha}{2z}}\right)^{\frac{p}{q}+1}\right) - \frac{p}{q}\mathrm{Tr}\left[\left(\sigma^{\frac{1-\alpha}{2z}}\rho^{\frac{\alpha}{z}}\sigma^{\frac{1-\alpha}{2z}}\right)^p\right] \\
     = & \left(p-\frac{p}{q}\right)\mathrm{Tr}\left[\left(\sigma^{\frac{1-\alpha}{2z}}\rho^{\frac{\alpha}{z}}\sigma^{\frac{1-\alpha}{2z}}\right)^p\right]
     \\
     = & \mathrm{Tr}\left[\left(\sigma^{\frac{1-\alpha}{2z}}\rho^{\frac{\alpha}{z}}\sigma^{\frac{1-\alpha}{2z}}\right)^p\right],
\end{align}
where $1=\frac{1}{p}+\frac{1}{q}$ implies that $p+q=pq$ which implies $\frac{p}{q}+1 = p$ and $p-\frac{p}{q}=1$. For a choice of $p=z$ and $q=\frac{z}{z-1}$, where $z>1$ define the left hand side of (\ref{newIneq1}) as 
\begin{equation}\label{function1}
    f_{\alpha,z}(H,\rho,\sigma) :=      z\mathrm{Tr}(\sigma^{\frac{z-\alpha}{2z}}\rho^{\frac{\alpha}{z}}\sigma^{\frac{z-\alpha}{2z}}H) - (z-1)\mathrm{Tr}\left[(\sigma^{\frac{z-1}{2z}}H\sigma^{\frac{z-1}{2z}})^{\frac{z}{z-1}}\right].
\end{equation} Then, the desired result is achieved.
\end{proof}

\begin{proposition}\label{jconvexity}
Let $1<\alpha\leq 2$ and $\frac{\alpha}{2}\leq z\leq\alpha$. For a fixed $H>0$, one has that $f_{\alpha,z}(H,\rho,\sigma)$ from equation (\ref{f1}) is jointly convex in $\rho$ and $\sigma$, where $\sigma$ is in any convex subset of $\mathcal{I}(\mathcal{H})$ from Remark \ref{remark: newset}. 
\end{proposition}

\begin{proof}
Fix $H>0$. Let $p = \frac{\alpha}{z}$, $q = \frac{\alpha-z}{z}$, $X=\rho^{\frac{p}{2}}$, and $Y = \sigma^{-\frac{q}{2}}H^{\frac{1}{2}}$. Note that \begin{align}
    \mathrm{Tr}\left(\sigma^{\frac{z-\alpha}{2z}}\rho^{\frac{\alpha}{z}}\sigma^{\frac{z-\alpha}{2z}}H\right) & = \mathrm{Tr}\left(H^{\frac{1}{2}}\sigma^{\frac{z-\alpha}{2z}}\rho^{\frac{\alpha}{z}}\sigma^{\frac{z-\alpha}{2z}}H^{\frac{1}{2}}\right) \\
    & = \mathrm{Tr}\left(H^{\frac{1}{2}}\sigma^{\frac{z-\alpha}{2z}}\rho^{\frac{\alpha}{2z}}\rho^{\frac{\alpha}{2z}}\sigma^{\frac{z-\alpha}{2z}}H^{\frac{1}{2}}\right) \\
    & = \mathrm{Tr}\left(H^{\frac{1}{2}}\sigma^{-\frac{q}{2}}\rho^{\frac{p}{2}}\rho^{\frac{p}{2}}\sigma^{-\frac{q}{2}}H^{\frac{1}{2}}\right)\\
    & = \mathrm{Tr}\left(Y^*X^*XY\right)\\
    & = \mathrm{Tr}|XY|^2.
\end{align}
Since $1<z\leq\alpha\leq 2z$, observe that $1\leq p = \frac{\alpha}{z}\leq 2$ and $0< q = \frac{\alpha-z}{z}\leq 1$. Both $p$ and $q$ are both positive numbers, so set $(r_0,r_1,r_2) = (\frac{2}{p},2,\frac{2}{q})$.

% fixed up until here

Then by Theorem \ref{zhangresult}, 
\begin{equation}\label{partialsup}
\mathrm{Tr}\left(\sigma^{\frac{z-\alpha}{2z}}\rho^{\frac{\alpha}{z}}\sigma^{\frac{z-\alpha}{2z}}H\right) = \max\left\{ p\mathrm{Tr}|\rho^{\frac{p}{2}}Z|^{\frac{2}{p}}-q\mathrm{Tr}|H^{-\frac{1}{2}}\sigma^{\frac{q}{2}}Z|^{\frac{2}{q}} \right\},
\end{equation}
where $Z$ is invertible. Since $1\leq p\leq 2$, it follows from Proposition \ref{fixedB} part 2 that 
$$\rho\mapsto \mathrm{Tr}|\rho^{\frac{p}{2}}Z|^{\frac{2}{p}} = \mathrm{Tr}\left[\left(Z^*\rho^p Z\right)^{\frac{1}{p}}\right] $$ is convex in $\rho$. 
For the second term in (\ref{partialsup}), choose $\Phi_Z(X):= Z^*XZ$, which is linear, positive, and self adjoint. Let $A = Z^{-1}H^{-1}(Z^*)^{-1}$, then by Proposition \ref{newPick}, 
\begin{align}
    \mathrm{Tr}|H^{-\frac{1}{2}}\sigma^{\frac{q}{2}}Z|^{\frac{2}{q}} & = \mathrm{Tr}\left(\left[(H^{-\frac{1}{2}}\sigma^{\frac{q}{2}}Z)^*(H^{-\frac{1}{2}}\sigma^{\frac{q}{2}}Z)\right]^{\frac{1}{q}}\right) \\
    & = \mathrm{Tr}\left(\left[Z^*\sigma^{\frac{q}{2}}H^{-1}\sigma^{\frac{q}{2}}Z\right]^{\frac{1}{q}}\right) \\
    & = \mathrm{Tr}\left(\left[Z^*\sigma^{\frac{q}{2}}ZZ^{-1}H^{-1}(Z^*)^{-1}Z^*\sigma^{\frac{q}{2}}Z\right]^{\frac{1}{q}}\right) \\
    & = \mathrm{Tr}\left[\left\{\Phi_Z\left(\sigma^{\frac{q}{2}}\right)A \,\Phi_Z\left(\sigma^{\frac{q}{2}}\right)\right\}^{\frac{1}{q}}\right]
\end{align} is concave in $\sigma$. Since $0<q$,
$$\sigma\mapsto -q\mathrm{Tr}|H^{-\frac{1}{2}}\sigma^{-\frac{q}{2}}Z|^{\frac{2}{q}} $$ is convex in $\sigma$. By Proposition \ref{sum}, one concludes that $\mathrm{Tr}\left(\sigma^{\frac{z-\alpha}{2z}}\rho^{\frac{\alpha}{z}}\sigma^{\frac{z-\alpha}{2z}}H\right)$ is the maximum of a sum of two convex functionals and thus is itself convex in $\rho$ and $\sigma$. 
$ $ 

%%%%%%%%%%%%%%%%%%%%%%%%%

On the other hand, since $\sigma^{\frac{z-1}{2z}}H\sigma^{\frac{z-1}{2z}}$ and $H^{\frac{1}{2}}\sigma^{\frac{z-1}{z}}H^{\frac{1}{2}}$ have the same nonzero eigenvalues,  
$$\mathrm{Tr}\left(\left[\sigma^{\frac{z-1}{2z}}H\sigma^{\frac{z-1}{2z}}\right]^{\frac{z}{z-1}}\right)=\mathrm{Tr}\left(\left[H^{\frac{1}{2}}\sigma^{\frac{z-1}{z}}H^{\frac{1}{2}}\right]^{\frac{z}{z-1}}\right).$$
By Proposition \ref{fixedB} part 1, $\mathrm{Tr}\left(\left[\sigma^{\frac{z-1}{2z}}H\sigma^{\frac{z-1}{2z}}\right]^{\frac{z}{z-1}}\right)$ is concave in $\sigma$, for all $z>1$. Hence $$-(z-1)\mathrm{Tr}\left(\left[\sigma^{\frac{z-1}{2z}}H\sigma^{\frac{z-1}{2z}}\right]^{\frac{z}{z-1}}\right)$$ is convex in $\sigma$. As a consequence, by Proposition \ref{sum},
$$ f_{\alpha,z}(H,\rho,\sigma)=      z\mathrm{Tr}\left(\sigma^{\frac{z-\alpha}{2z}}\rho^{\frac{\alpha}{z}}\sigma^{\frac{z-\alpha}{2z}}H\right) - (z-1)\mathrm{Tr}\left(\left[\sigma^{\frac{z-1}{2z}}H\sigma^{\frac{z-1}{2z}}\right]^{\frac{z}{z-1}}\right) $$ is a sum of two convex functionals in $\rho$ and $\sigma$ as desired. 
\end{proof}

%%%%%%%%%%%%%%% Supporting Result %%%%%%%%%%
%\section{Preliminaries}
%\input{prelims.tex}

%%%%%%%%%%%%%%% Main Result %%%%%%%%%%
\section{Main Result}\label{sec: main}

\begin{remark}
Let us denote $\mathcal{Q}(\mathcal{H}_{AB})$ as
\begin{equation}
    \cQ(\cH_{AB}):= \{X_{AB}\in\cI(\cH_{AB}) : X_A\otimes\pi_B\in\cI(\cH_{AB})\},
\end{equation}
where $\pi_B$ is the maximally mixed state on Hilbert space $\cH_B$.
\end{remark}

The techniques used here are inspired by \cite{sandwichConditions} and \cite{CFLconjecture}. The main result of this section is a consequence of the following theorem. 

\begin{theorem}[Necessary Partial Trace Case]\label{partial case1}
Let $\rho_{AB}\in\cD(\cH_{AB})$ and $\sigma_{AB}\in\cQ(\cH_{AB})$. For any $1<\alpha\leq 2$ and $\frac{\alpha}{2}\leq z\leq\alpha$, whenever saturation of the DPI holds. i.e., $D_{\alpha,z}(\rho_{AB}||\sigma_{AB})=D_{\alpha,z}(\rho_{A}||\sigma_{A})$, then the states satisfy 
\begin{equation}
\sigma_{AB}^{\frac{1-z}{2z}}\left(\sigma_{AB}^{\frac{1-\alpha}{2z}}\rho_{AB}^{\frac{\alpha}{z}}\sigma_{AB}^{\frac{1-\alpha}{2z}}\right)^{z-1}\sigma_{AB}^{\frac{1-z}{2z}}=\sigma_A^{\frac{1-z}{2z}}\left(\sigma_A^{\frac{1-\alpha}{2z}}\rho_A^{\frac{\alpha}{z}}\sigma_A^{\frac{1-\alpha}{2z}}\right)^{z-1}\sigma_A^{\frac{1-z}{2z}}.
\end{equation}
\end{theorem}

\begin{proof}
Assume $1<\alpha\leq 2$ and $\frac{\alpha}{2}\leq z\leq\alpha$, denote $d=\mbox{dim}(\cH_B)$, and define
$$\rho_i:=(\mathbbm{1}\otimes v_i)\rho_{AB}(\mathbbm{1}\otimes v_i^*),$$
where $\rho_{AB}\in\cD(\cH_{AB})$ and $\{v_i\}_{i=1}^{d^2}$ are the generalized Pauli \footnote{For more details on the generalized Pauli matrices, see \cite{wildeBook} Chapter 3.7 for more details.} matrices. Similarly, define 
$$\sigma_i:=(\mathbbm{1}\otimes v_i)\sigma_{AB}(\mathbbm{1}\otimes v_i^*),$$
where $\sigma_{AB}\in\cI(\cH_{AB})$. Then $\sigma_i\in\cI(\cH_{AB})$, for all $i$. Define $\lambda_i=\frac{1}{d^2}$, for all $i=1,\dots,d^2$, and let 
$$\tilde{\rho}=\sum\limits_{i=1}^{d^2}\lambda_i\rho_i \mbox{ and } \tilde{\sigma}=\sum\limits_{i=1}^{d^2}\lambda_i\sigma_i.$$ 
As mentioned in \cite{sandwichConditions}, 
$$\tilde{\rho}=\sum\limits_{i=1}^{d^2}\lambda_i\rho_i=\rho_A\otimes \pi_B \text{  and  }\tilde{\sigma}=\sum\limits_{i=1}^{d^2}\lambda_i\sigma_i=\sigma_A\otimes\pi_B,$$
where $\pi_B$ is the completely mixed state \footnote{See \cite{wildeBook} exercise 4.7.6 for an explanation of how the Generalized Pauli operators randomly applied to any density operator with uniform probability give us a maximally mixed state} on $\cH_B$. i.e., $\pi_B=\frac{\mathbbm{1}}{d}$. Note that $$\tilde{\sigma} = \sigma_A\otimes\pi_B\in\cQ(\cH_{AB})\subseteq\cI(\cH_{AB}).$$
Define
\begin{equation}\label{argMax1}
    \bar{H}:=\argmax_{H>0}f_{\alpha,z}(H,\tilde{\rho},\tilde{\sigma}) \mbox{ and } H_i:=\argmax_{H>0}f_{\alpha,z}(H,\tilde{\rho_i},\tilde{\sigma_i}), \text{ where}
\end{equation}
\begin{equation}
    f_{\alpha,z}(H,\rho,\sigma) :=      z\mathrm{Tr}(\sigma^{\frac{z-\alpha}{2z}}\rho^{\frac{\alpha}{z}}\sigma^{\frac{z-\alpha}{2z}}H) - (z-1)\mathrm{Tr}\left[(\sigma^{\frac{z-1}{2z}}H\sigma^{\frac{z-1}{2z}})^{\frac{z}{z-1}}\right],
\end{equation}
from Lemma \ref{function2}.
Note that the trace functional $\Psi_{\alpha,z}(\rho||\sigma)$, mentioned after theorem \ref{alphazDPI} is proven jointly convex in \cite{CFLconjecture} in a more general setting. The following chain of inequalities says:
\begin{align}
    \Psi_{\alpha,z}(\tilde{\rho}||\tilde{\sigma})&=f_{\alpha,z}(\bar{H},\tilde{\rho},\tilde{\sigma})\label{11}\\
    &  \leq \sum\limits_{i=1}^{d^2}\lambda_i f_{\alpha,z}(\bar{H},\rho_i,\sigma_i) \label{21} \\
    & \leq \sum\limits_{i=1}^{d^2}\lambda_i f_{\alpha,z}(H_i,\rho_i,\sigma_i) \label{31} \\
    & =\sum\limits_{i=1}^{d^2}\lambda_i \Psi_{\alpha,z}(\rho_i||\sigma_i) \label{41},
\end{align} 
where lines (\ref{11}) and (\ref{41}) are from Lemma \ref{function2}, line (\ref{21}) is from the joint convexity of $f$, which is proven in Proposition \ref{jconvexity}, and line (\ref{31}) is from the fact that $H_i$ is the maximizer for $f_{\alpha,z}(H_i,\rho_i,\sigma_i)$ from (\ref{argMax1}). Assuming saturation of the DPI is equivalent to $$\Psi_{\alpha,z}(\tilde{\rho}||\tilde{\sigma})=\sum\limits_{i=1}^{d^2}\lambda_i\Psi_{\alpha,z}(\rho_i||\sigma_i).$$
Then the chain of inequalities above is now a chain of equalities and thus by the definition of $H_i$, 
$$f_{\alpha,z}(\bar{H},\rho_i,\sigma_i)=f_{\alpha,z}(H_i,\rho_i,\sigma_i) \mbox{ for all } i=1,\dots,d^2.$$
By the uniqueness of the maximizer $\bar{H}$, which is proven above in Lemma \ref{function2},  $\bar{H}=H_i$ for all $i=1,\dots,d^2$. 

Recall that because an operator $X$ and $UXU^*$ have the same eigenvalues, where $U$ is any unitary, then for any function $f$ it follows that $f(UXU^*)=Uf(X)U^*$. Therefore from Lemma \ref{function2}, one has an explicit form of the maximizer: $H_i=$ \begin{align}
     &\sigma_i^{\frac{1-z}{2z}}\left(\sigma_i^{\frac{1-\alpha}{2z}}\rho_i^{\frac{\alpha}{z}}\sigma_i^{\frac{1-\alpha}{2z}}\right)^{z-1}\sigma_i^{\frac{1-z}{2z}} \\
     = &[(\mathbbm{1}\otimes v_i)\sigma_{AB}(\mathbbm{1}\otimes v_i^*)]^{\frac{1-z}{2z}}\\
    &\left([(\mathbbm{1}\otimes v_i)\sigma_{AB}(\mathbbm{1}\otimes v_i^*)]^{\frac{1-\alpha}{2z}}[(\mathbbm{1}\otimes v_i)\rho_{AB}(\mathbbm{1}\otimes v_i^*)]^{\frac{\alpha}{z}}[(\mathbbm{1}\otimes v_i)\sigma_{AB}(\mathbbm{1}\otimes v_i^*)]^{\frac{1-\alpha}{2z}}\right)^{z-1}\\
    &[(\mathbbm{1}\otimes v_i)\sigma_{AB}(\mathbbm{1}\otimes v_i^*)]^{\frac{1-z}{2z}} \\
    = & (\mathbbm{1}\otimes v_i)\sigma_{AB}^{\frac{1-z}{2z}}\left(\sigma_{AB}^{\frac{1-\alpha}{2z}}\rho_{AB}^{\frac{\alpha}{z}}\sigma_{AB}^{\frac{1-\alpha}{2z}}\right)^{z-1}\sigma_{AB}^{\frac{1-z}{2z}}(\mathbbm{1}\otimes v_i^*). 
\end{align}
This holds for all $v_i$ due to the fact that $v_i^*v_i=I$. Therefore for some $i\in\{1,\dots,d^2\}$, 
$$H_i=\sigma_{AB}^{\frac{1-z}{2z}}\left(\sigma_{AB}^{\frac{1-\alpha}{2z}}\rho_{AB}^{\frac{\alpha}{z}}\sigma_{AB}^{\frac{1-\alpha}{2z}}\right)^{z-1}\sigma_{AB}^{\frac{1-z}{2z}}.$$ 
Also by similar calculations, $\bar{H}=$
\begin{align}
     &  (\sigma_A\otimes\pi_B)^{\frac{1-z}{2z}}\left((\sigma_A\otimes\pi_B)^{\frac{1-\alpha}{2z}}(\rho_A\otimes\pi_B)^{\frac{\alpha}{z}}(\sigma_A\otimes\pi_B)^{\frac{1-\alpha}{2z}}\right)^{z-1}(\sigma_A\otimes\pi_B)^{\frac{1-z}{2z}}\\
    & = \sigma_A^{\frac{1-z}{2z}}\left(\sigma_A^{\frac{1-\alpha}{2z}}\rho_A^{\frac{\alpha}{z}}\sigma_A^{\frac{1-\alpha}{2z}}\right)^{z-1}\sigma_A^{\frac{1-z}{2z}} \otimes \pi_B^{\left(\frac{1-z}{2z}+(\frac{1-\alpha}{2z}+\frac{\alpha}{z}+\frac{1-\alpha}{2z})(z-1)+\frac{1-z}{2z}\right)} \\
    & = \sigma_A^{\frac{1-z}{2z}}\left(\sigma_A^{\frac{1-\alpha}{2z}}\rho_A^{\frac{\alpha}{z}}\sigma_A^{\frac{1-\alpha}{2z}}\right)^{z-1}\sigma_A^{\frac{1-z}{2z}} \otimes \mathbbm{1}_B \label{zero1},
\end{align}
where (\ref{zero1}) holds because $\frac{1-z}{2z}+(\frac{1-\alpha}{2z}+\frac{\alpha}{z}+\frac{1-\alpha}{2z})(z-1)+\frac{1-z}{2z} = 0$. 
Thus 
\begin{equation}
\sigma_{AB}^{\frac{1-z}{2z}}\left(\sigma_{AB}^{\frac{1-\alpha}{2z}}\rho_{AB}^{\frac{\alpha}{z}}\sigma_{AB}^{\frac{1-\alpha}{2z}}\right)^{z-1}\sigma_{AB}^{\frac{1-z}{2z}}=\sigma_A^{\frac{1-z}{2z}}\left(\sigma_A^{\frac{1-\alpha}{2z}}\rho_A^{\frac{\alpha}{z}}\sigma_A^{\frac{1-\alpha}{2z}}\right)^{z-1}\sigma_A^{\frac{1-z}{2z}}.
\end{equation} 
\end{proof}

Next is the generalization of the partial case trace using a standard Stinespring Dilation argument. 
\begin{corollary}\label{full case}
Let $\rho\in\cD(\cH)$, $\sigma\in\mathcal{Q}(\mathcal{H})$, and $\Lambda:\cB(\cH)\to\cB(\cK)$ be a quantum channel. For any $1<\alpha\leq 2$ and $\frac{\alpha}{2}\leq z\leq\alpha$,  whenever saturation of the DPI holds, i.e., $D_{\alpha,z}(\rho||\sigma)=D_{\alpha,z}(\Lambda(\rho)||\Lambda(\sigma))$,   then the states satisfy  
\begin{equation}
\sigma^{\frac{1-z}{2z}}(\sigma^{\frac{1-\alpha}{2z}}\rho^{\frac{\alpha}{z}}\sigma^{\frac{1-\alpha}{2z}})^{z-1}\sigma^{\frac{1-z}{2z}}=\Lambda^*\left(\Lambda(\sigma)^{\frac{1-z}{2z}}(\Lambda(\sigma)^{\frac{1-\alpha}{2z}}\rho^{\frac{\alpha}{z}}\Lambda(\sigma)^{\frac{1-\alpha}{2z}})^{z-1}\Lambda(\sigma)^{\frac{1-z}{2z}}\right).
\end{equation}
\end{corollary}

\begin{proof}[Proof of \ref{full case}]Following \cite{sandwichConditions}, for any quantum channel $\Lambda:\cB(\cH)\to\cB(\cK)$, by the Stinespring Dilation Theorem \footnote{Stinespring Dilation Theorem can be found in \cite{wildeBook}.},  there exists a Hilbert space $\cH'$, a pure state $\ket{\tau}\in\cH'\otimes\cK$, and a unitary operator $U:\cH\otimes\cH'\otimes\cK\to \cH\otimes\cH'\otimes\cK$ such that for every $\rho\in\cB(\cH)$, one has
$$\Lambda(\rho)=\mathrm{Tr}_{12}\left(U(\rho\otimes\tau)U^*\right), $$ 
where $\tau=\ket{\tau}\bra{\tau}$ and $\mathrm{Tr}_{12}$ denotes the partial trace over the first two systems $\cH\otimes\cH'$ i.e., $\mathrm{Tr}_{12}:\cH\otimes\cH'\otimes\cK\to\cK$.  Then, for the parameters \footnote{for $\alpha$ and $z$ where the DPI makes sense. see Theorem \ref{alphazDPI} for such parameters} where DPI is satisfied for the $\alpha-z$ RRE, 
\begin{align}
    D_{\alpha,z}(\rho||\sigma) = & D_{\alpha,z}(U(\rho\otimes\tau)U^*||U(\sigma\otimes\tau)U^*)\label{properties} \\
    \geq & D_{\alpha,z}(\mathrm{Tr}_{12}(U(\rho\otimes\tau)U^*)||\mathrm{Tr}_{12}(U(\sigma\otimes\tau)U^*)) \label{partialCase} \\
    = & D_{\alpha,z}(\Lambda(\rho)||\Lambda(\sigma)),
\end{align}
where (\ref{properties}) is due to properties mentioned in \ref{alphaZ} and (\ref{partialCase}) is the DPI for partial traces. By assuming equality and by Theorem \ref{partial case1} one sees that 
\begin{align}
    &\mathbbm{1}_{\cH\otimes\cH'}\otimes\Lambda(\sigma)^{\frac{1-z}{2z}}\left(\Lambda(\sigma)^{\frac{1-\alpha}{2z}}\Lambda(\rho)^{\frac{\alpha}{z}}\Lambda(\sigma)^{\frac{1-\alpha}{2z}}\right)^{z-1}\Lambda(\sigma)^{\frac{1-z}{2z}}\cdot \\
   = & [U(\sigma\otimes\tau)U^*]^{\frac{1-z}{2z}} \\
   &\left([U(\sigma\otimes\tau)U^*]^{\frac{1-\alpha}{2z}}[U(\rho\otimes\tau)U^*]^{\frac{\alpha}{z}}[U(\sigma\otimes\tau)U^*]^{\frac{1-\alpha}{2z}}\right)^{z-1}[U(\sigma\otimes\tau)U^*]^{\frac{1-z}{2z}} \\
   = & U(\sigma^{\frac{1-z}{2z}}\left(\sigma^{\frac{1-\alpha}{2z}}\rho^{\frac{\alpha}{z}}\sigma^{\frac{1-\sigma}{2z}}\right)^{z-1}\sigma^{\frac{1-z}{2z}}\otimes\tau)U^* \label{preAdjoint},
\end{align}
where the last line is due to the fact that $f(UXU^*)=Uf(X)U^*$, for every function $f$ and for any unitary $U$. 
The quantum channel $\Lambda:\cB(\cH)\to\cB(\cK)$, has a unique adjoint, $\Lambda^*:\cB(\cH)\to\cB(\cK)$, and it is given by
\begin{equation}\label{adjoint}
\Lambda^*(X):=(\mathbbm{1}_\cH\otimes\bra{\tau}_{\cH'\otimes\cK})U^*(X)U(\mathbbm{1}_\cH\otimes\ket{\tau}_{\cH'\otimes\cK}).
\end{equation} 
 So applying (\ref{adjoint}) to (\ref{preAdjoint}) gives
 \begin{align}
     \Lambda^*\left[\Lambda(\sigma)^{\frac{1-z}{2z}}\left(\Lambda(\sigma)^{\frac{1-\alpha}{2z}}\Lambda(\rho)^{\frac{\alpha}{z}}\Lambda(\sigma)^{\frac{1-\alpha}{2z}}\right)^{z-1}\Lambda(\sigma)^{\frac{1-z}{2z}}\right] &=\\
     \Lambda^*\left[U\left(\sigma^{\frac{1-z}{2z}}\left(\sigma^{\frac{1-\alpha}{2z}}\rho^{\frac{\alpha}{z}}\sigma^{\frac{1-\alpha}{2z}}\right)^{z-1}\sigma^{\frac{1-z}{2z}}\otimes\tau\right)U^*\right] &= \\
     (\mathbbm{1}_{\cH}\otimes\bra{\tau})\left(\sigma^{\frac{1-z}{2z}}\left(\sigma^{\frac{1-\alpha}{2z}}\rho^{\frac{\alpha}{z}}\sigma^{\frac{1-\alpha}{2z}}\right)^{z-1}\sigma^{\frac{1-z}{2z}}\otimes\tau\right)(\mathbbm{1}_{\cH}\otimes\ket{\tau})&= \\
     \sigma^{\frac{1-z}{2z}}\left(\sigma^{\frac{1-\alpha}{2z}}\rho^{\frac{\alpha}{z}}\sigma^{\frac{1-\alpha}{2z}}\right)^{z-1}\sigma^{\frac{1-z}{2z}}.
 \end{align}
\end{proof}

Here, some algebraic sufficient conditions for saturating the DPI of $\alpha-z$ RRE are explained. 

\begin{proposition}\label{algebraicsufficient}
If $\sigma_{AB}\in\cP(\cH_{AB})$ and $\rho_{AB}\in\cD(\cH_{AB})$, where $\rho_{AB}$ is a product state such that 
$$ \sigma_{AB}^{\frac{1-\alpha}{2z}}\left(\sigma_{AB}^{\frac{1-\alpha}{2z}}\rho_{AB}^{\frac{\alpha}{z}}\sigma_{AB}^{\frac{1-\alpha}{2z}}\right)^{z-1}\sigma_{AB}^{\frac{1-\alpha}{2z}}=\sigma_A^{\frac{1-\alpha}{2z}}\left(\sigma_A^{\frac{1-\alpha}{2z}}\rho_A^{\frac{\alpha}{z}}\sigma_A^{\frac{1-\alpha}{2z}}\right)^{z-1}\sigma_A^{\frac{1-\alpha}{2z}} \otimes\mathbbm{1}_B,$$ 
then $$D_{\alpha,z}(\rho_{AB}||\sigma_{AB})=D_{\alpha,z}(\rho_A||\sigma_A) + k_{\rho},$$
where  $k_\rho= \frac{1}{\alpha-1}\log\left[\mathrm{Tr}\left(\rho_B^{\frac{\alpha}{z}}\right)\right]$. 
\end{proposition}

\begin{proof}
Since $\rho_{AB}$ is a product state on $\cH_{AB}$, it follows that $\rho_{AB}=\rho_A\otimes\rho_B$, where $\rho_A\in\mathcal{D}(\mathcal{H}_A)$ and $\rho_B\in\mathcal{D}(\mathcal{H}_B)$. Multiplying the assumed expression by $\rho_{AB}$ on the left and taking a trace gives \footnote{Note that we are using the fact that $\mathrm{Tr}_{AB}(X)=\mathrm{Tr}_A(\mathrm{Tr}_BX)$}

$$ \mathrm{Tr}\left(\rho_{AB}^\frac{\alpha}{z}\sigma_{AB}^{\frac{1-\alpha}{2z}}\left(\sigma_{AB}^{\frac{1-\alpha}{2z}}\rho_{AB}^{\frac{\alpha}{z}}\sigma_{AB}^{\frac{1-\alpha}{2z}}\right)^{z-1}\sigma_{AB}^{\frac{1-\alpha}{2z}}\right) =
\mathrm{Tr}\left(\rho_{AB}^\frac{\alpha}{z}\left(\sigma_A^{\frac{1-\alpha}{2z}}\left(\sigma_A^{\frac{1-\alpha}{2z}}\rho_A^{\frac{\alpha}{z}}\sigma_A^{\frac{1-\alpha}{2z}}\right)^{z-1}\sigma_A^{\frac{1-\alpha}{2z}} \otimes\mathbbm{1}_B \right)\right).$$ This implies 
$$ \mathrm{Tr}\left[\left(\sigma_{AB}^{\frac{1-\alpha}{2z}}\rho_{AB}^{\frac{\alpha}{z}}\sigma_{AB}^{\frac{1-\alpha}{2z}}\right)^z\right] = \mathrm{Tr} \left((\rho_A^{\frac{\alpha}{z}}\otimes\rho_B^{\frac{\alpha}{z}})\sigma_A^{\frac{1-\alpha}{2z}}\left(\sigma_A^{\frac{1-\alpha}{2z}}\rho_A^{\frac{\alpha}{z}}\sigma_A^{\frac{1-\alpha}{2z}}\right)^{z-1}\sigma_A^{\frac{1-\alpha}{2z}}\right),$$
 which gives
 $$ \mathrm{Tr}\left[\left(\sigma_{AB}^{\frac{1-\alpha}{2z}}\rho_{AB}^{\frac{\alpha}{z}}\sigma_{AB}^{\frac{1-\alpha}{2z}}\right)^z\right] = \mathrm{Tr}\left[ \left(\rho_A^{\frac{\alpha}{z}}\sigma_A^{\frac{1-\alpha}{2z}}\left(\sigma_A^{\frac{1-\alpha}{2z}}\rho_A^{\frac{\alpha}{z}}\sigma_A^{\frac{1-\alpha}{2z}}\right)^{z-1}\sigma_A^{\frac{1-\alpha}{2z}}\right)\otimes \rho_B^{\frac{\alpha}{z}}\right].$$ Taking the $\log$ of both sides and multiplying by $\frac{1}{\alpha-1}$ gives 
 $$ \frac{1}{\alpha-1}\log\left(\mathrm{Tr}\left[\left(\sigma_{AB}^{\frac{1-\alpha}{2z}}\rho_{AB}^{\frac{\alpha}{z}}\sigma_{AB}^{\frac{1-\alpha}{2z}}\right)^z\right]\right) = \frac{1}{\alpha-1}\log\left(\mathrm{Tr}\left[ \left(\sigma_A^{\frac{1-\alpha}{2z}}\rho_A^{\frac{\alpha}{z}}\sigma_A^{\frac{1-\alpha}{2z}}\right)^{z}\right]\mathrm{Tr}\left( \rho_B^{\frac{\alpha}{z}}\right)\right),$$ which is the same as 
 $$D_{\alpha,z}(\rho_{AB}||\sigma_{AB}) = D_{\alpha,z}(\rho_A||\sigma_A) + k_\rho, 
 $$ 
 where  $k_\rho= \frac{1}{\alpha-1}\log\left[\mathrm{Tr}\left(\rho_B^{\frac{\alpha}{z}}\right)\right]$ as desired.
  
\end{proof}

\begin{remark}
It is interesting to see that Theorem \ref{partial case1} and Proposition \ref{algebraicsufficient} would hold simultaneously if and only if $z=\alpha$. This in turn will result in $\alpha-$SRD, which aligns with the work done in \cite{sandwichConditions}.
\end{remark}

\begin{remark}
If $\rho_B$ in Proposition \ref{algebraicsufficient} is a pure state, then $k=0$. This immediately leads to another result:
\end{remark}

\begin{proposition}
If $\sigma_{AB}\in\cP(\cH_{AB})$ and $\rho_{AB}\in\cD(\cH_{AB})$, where $\rho_{AB}$ is a separable state, such that 
$$ \sigma_{AB}^{\frac{1-\alpha}{2z}}\left(\sigma_{AB}^{\frac{1-\alpha}{2z}}\rho_{AB}^{\frac{\alpha}{z}}\sigma_{AB}^{\frac{1-\alpha}{2z}}\right)^{z-1}\sigma_{AB}^{\frac{1-\alpha}{2z}}=\sigma_A^{\frac{1-\alpha}{2z}}\left(\sigma_A^{\frac{1-\alpha}{2z}}\rho_A^{\frac{\alpha}{z}}\sigma_A^{\frac{1-\alpha}{2z}}\right)^{z-1}\sigma_A^{\frac{1-\alpha}{2z}},$$ 
then the DPI under partial traces is saturated. i.e., $$D_{\alpha,z}(\rho_{AB}||\sigma_{AB})=D_{\alpha,z}(\rho_A||\sigma_A).$$
\end{proposition}

\begin{proof}
Since $\rho_{AB}$ is a separable state on $\cH_{AB}$, write $\rho_{AB}$ as a convex combination of a tensor product of pure states. i.e., $\rho_{AB}=\sum\limits_{i\in I}\lambda_i\ket{\psi_i}_A\bra{\psi_i}\otimes\ket{\phi_i}_B\bra{\phi_i}$, where $\{\ket{\psi_i}\}$ and $\{\ket{\phi_i}\}$ are sets of pure states on $\cH_A$ and $\cH_B$ respectively. Furthermore, $0\leq\lambda_i\leq1$, for all $i\in I$ so that $\sum\limits_{i\in I}\lambda_i=1$. Again, multiplying by $\rho_{AB}$ on the left and taking a trace of the assumed expression yields 
$$  \mathrm{Tr}\left[\left(\sigma_{AB}^{\frac{1-\alpha}{2z}}\rho_{AB}^{\frac{\alpha}{z}}\sigma_{AB}^{\frac{1-\alpha}{2z}}\right)^z\right] = \mathrm{Tr} \left(\left(\sum_{i\in I}\lambda_i\ket{\psi_i}_A\bra{\psi_i}\otimes\ket{\phi_i}_B\bra{\phi_i}\right)^{\frac{\alpha}{z}}\sigma_A^{\frac{1-\alpha}{2z}}\left(\sigma_A^{\frac{1-\alpha}{2z}}\rho_A^{\frac{\alpha}{z}}\sigma_A^{\frac{1-\alpha}{2z}}\right)^{z-1}\sigma_A^{\frac{1-\alpha}{2z}}\right).$$
This implies 
$$ \mathrm{Tr}\left[\left(\sigma_{AB}^{\frac{1-\alpha}{2z}}\rho_{AB}^{\frac{\alpha}{z}}\sigma_{AB}^{\frac{1-\alpha}{2z}}\right)^z\right] = $$
$$\mathrm{Tr}\left[ \left(\left(\sum_{i\in I}\lambda_i^{\frac{\alpha}{z}}\ket{\psi_i}_A\bra{\psi_i}\right)\sigma_A^{\frac{1-\alpha}{2z}}\left(\sigma_A^{\frac{1-\alpha}{2z}}\rho_A^{\frac{\alpha}{z}}\sigma_A^{\frac{1-\alpha}{2z}}\right)^{z-1}\sigma_A^{\frac{1-\alpha}{2z}}\right)\otimes\ket{\phi_i}_B\bra{\phi_i}\right], $$
which gives
$$\mathrm{Tr}\left[\left(\sigma_{AB}^{\frac{1-\alpha}{2z}}\rho_{AB}^{\frac{\alpha}{z}}\sigma_{AB}^{\frac{1-\alpha}{2z}}\right)^z\right] = \mathrm{Tr}\left[ \left(\rho_A^{\frac{\alpha}{z}}\sigma_A^{\frac{1-\alpha}{2z}}\left(\sigma_A^{\frac{1-\alpha}{2z}}\rho_A^{\frac{\alpha}{z}}\sigma_A^{\frac{1-\alpha}{2z}}\right)^{z-1}\sigma_A^{\frac{1-\alpha}{2z}}\right)\right].$$
Taking the $\log$ of both sides and multiplying by $\frac{1}{\alpha-1}$ gives 
$$ D_{\alpha,z}(\rho_{AB}||\sigma_{AB}) = D_{\alpha,z}(\rho_A||\sigma_A). $$
\end{proof}

%%%%%%%%%%%%%%% Closing Remarks %%%%%%%%%%
\section{Closing Remarks}

We have shown algebraic conditions equivalent to saturating the data processing inequality for $1<z\leq\alpha\leq2z$, which generalizes the $\alpha-$ SRD saturation condition from \cite{sandwichConditions}. The techniques in this paper fail for $\alpha<1$ because the H\"{o}lder inequality requires positive powers, so it would be interesting to find a similar result for this case. As mentioned in section \ref{sec: notations}, a quantum channel $\Lambda$ is said to be \emph{sufficient} with respect to $\rho$ and $\sigma$ if there exists a quantum channel $\cR$ such that $(\cR\circ\Lambda)(\rho)=\rho$ and $(\cR\circ\Lambda)(\sigma)=\sigma$. For Umegaki relative entropy, $\alpha$-RRE, and $\alpha$-SRD, saturation of the DPI is equivalent to sufficiency of the quantum channel $\Lambda$. In general, it is not known whether sufficiency of the quantum channel is equivalent to saturation of the $\alpha-z$ RRE DPI. However in \cite{hiaimosonyi}, Hiai and Mosonyi do prove such results for a set of density operators fixed under the quantum channel. (i.e., $\Lambda(\rho)=\rho$ and $\Lambda(\sigma)=\sigma)$. It would be interesting to find a larger class of channels where sufficiency holds. 
\subsection{Acknowledgments}
The author is appreciative to Haonan Zhang for carefully reading this manuscript and asking questions that lead to the modification of Proposition \ref{newPick}. The author is also grateful to Anna Vershynina, who is her advisor, for her comments and suggestions throughout the entirety of this paper. This work is supported by the National Science Foundation (NSF) grant DMS-1812734. 

%%%%%%%%%%%%%%% BIB %%%%%%%%%%
%plain means standard ordering:
\bibliographystyle{plain}
\bibliography{bib}

\end{document}